\documentclass[11pt,a4paper]{article}
\usepackage{fullpage}

\usepackage[utf8]{inputenc}
\usepackage{amsthm,amssymb}
\usepackage{amsmath}
\makeatletter
\newcommand{\leqnomode}{\tagsleft@true\let\veqno\@@leqno}
\makeatother

\usepackage{authblk}
\usepackage{comment}
\usepackage{algorithm}
\usepackage{algpseudocode}
\usepackage{tcolorbox}
\usepackage{mathtools,thmtools}
\usepackage[mathcal,mathscr]{eucal}
\usepackage{epsfig,graphicx,graphics,color}
\usepackage{caption}
\usepackage{subcaption}
\usepackage{enumerate}
\usepackage{enumitem, crossreftools}
\usepackage[sort]{cite}
\usepackage{titling}
\usepackage{algorithm}
\usepackage{tikz}
\usepackage{hyperref}

\hypersetup{
    colorlinks=true,       
    linkcolor=blue,        
    citecolor=red,         
    filecolor=magenta,     
    urlcolor=cyan,         
    linktocpage=true
}

\newcounter{algsubstate}


\newtcolorbox{probbox}{arc=6pt,
                      colback=white!100,
                      colframe=black!50,
                      before skip=6pt,
                      after skip=6pt,
                      boxsep=1pt,
                      left=6pt,
                      right=6pt,
                      top=4pt,
                      bottom=4pt}

\newcommand{\decprob}[3]{
   \begin{center}%
    \begin{minipage}{0.96\linewidth}%
      \begin{probbox}
      \textsc{#1}\\[0.2ex]
      \textbf{Input:} #2\\[0.2ex]
      \textbf{Question:} #3
      \end{probbox}
    \end{minipage}%
  \end{center}
}

\theoremstyle{plain}
\newtheorem{theorem}{Theorem}
\newtheorem{lemma}[theorem]{Lemma}
\newtheorem{corollary}[theorem]{Corollary}

\newtheorem{qu}{Question}

\theoremstyle{definition}

\newtheorem{ex}[theorem]{Example}

\newtheorem{remark}[theorem]{Remark}

\newcommand{\cB}{\mathcal{B}}
\newcommand{\cC}{\mathcal{C}}

\newcommand{\cF}{\mathcal{F}}
\newcommand{\cG}{\mathcal{G}}
\newcommand{\cI}{\mathcal{I}}
\newcommand{\cJ}{\mathcal{J}}
\newcommand{\cH}{\mathcal{H}}
\newcommand{\cK}{\mathcal{K}}
\newcommand{\cM}{\mathcal{M}}

\newcommand{\cP}{\mathcal{P}}

\newcommand{\cT}{\mathcal{T}}

\newcommand{\FC}{\mathbb{T}}
\newcommand{\IS}{\mathbb{I}}
\newcommand{\PP}{\mathbb{P}}

\newcommand{\hmk}{\mathcal{K}^{-1}}
\newcommand{\hmt}{\mathcal{M}^{-1}}

\def\ox{\overline{x}}

\def\bx{\mathbf{x}}

\def\final{0}  
\ifnum\final=0  
\newcommand{\krnote}[1]{{\color{red}[{\tiny Krist\'of: \bf #1}]\marginpar{\color{red}*}}}
\newcommand{\enote}[1]{{\color{blue}[{\tiny Endre: \bf #1}]\marginpar{\color{blue}*}}}
\newcommand{\knote}[1]{{\color{orange}[{\tiny Kaz: \bf #1}]\marginpar{\color{red}*}}}
\else 
\newcommand{\krnote}[1]{}
\newcommand{\enote}[1]{}
\newcommand{\knote}[1]{}
\fi

\title{Hypergraph Horn Functions}

\author{
	Krist{\'o}f B{\'e}rczi\thanks{MTA-ELTE Momentum Matroid Optimization Research Group and MTA-ELTE Egerv\'ary Research Group, Department of Operations Research, E{\"o}tv{\"o}s Lor{\'a}nd University, Budapest, Hungary. 
	Email: \texttt{kristof.berczi@ttk.elte.hu}.} 
\and
	Endre Boros\thanks{MSIS Department and RUTCOR, Rutgers University, Piscataway, New Jersey, USA. Email: \texttt{endre.boros@rutgers.edu}.}
\and
	Kazuhisa Makino\thanks{Research Institute for Mathematical Sciences (RIMS) Kyoto University, Kyoto, Japan. Email: \texttt{makino@kurims.kyoto.ac.jp}.}
}

\begin{document}
\date{}
\maketitle

\begin{abstract}
Horn functions form a subclass of Boolean functions possessing interesting structural and computational properties. These functions play a fundamental role in algebra, artificial intelligence, combinatorics, computer science, database theory, and logic. 

In the present paper, we introduce the subclass of \emph{hypergraph Horn} functions that generalizes matroids and equivalence relations. We provide multiple characterizations of hypergraph Horn functions in terms of implicate-duality and the closure operator, which are respectively regarded as generalizations of matroid duality and Mac Lane\,--\,Steinitz exchange property of matroid closure. We also study algorithmic issues on hypergraph Horn functions, and show that the recognition problem (i.e., deciding if a given definite Horn CNF represents a hypergraph Horn function) and key realization (i.e., deciding if a given hypergraph is realized as a key set by a hypergraph Horn function) can be done in polynomial time, while implicate sets can be generated with polynomial delay. 
  \bigskip

\end{abstract}

\section{Introduction}
\label{sec:introduction}

In this paper we introduce and study a special subclass of Horn functions that is rich in mathematical structure and has advantageous algorithmic properties. 

Horn logic is a segment of propositional logic that came to prominence with the advent of computers, relational databases, logic programming, see e.g., 
\cite{McKinsey1943,Horn1951,CLM81,CL73,SDPF81,Fag82,Hod93,IM87,Mai83,MW88,Mak87}.
Horn expressions and functions appear in many equivalent forms in mathematics and computers science, such as directed hypergraphs, implications systems, 
closure operators, functional dependencies in relational databases, formal concept analysis, etc. We refer the reader to \cite{adaricheva_et_al:DR:2014:4619} for a more comprehensive overview of these equivalent areas and related literature. 

In this paper we use conjunctive normal form (CNF) representations for Boolean functions, and view Horn clauses equivalently as simple implications. For instance if $x_i$, $i=1,\dots,n$ are propositional variables (that is logical variables that can only take false/true values), then $\{x_1,x_2\}\to x_3$ is such a simple implication that can also be represented as a clause (an elementary disjunction) of the form $(\ox_1\vee \ox_2\vee x_3)$, where $\ox_i$ denotes the logical negation of $x_i$, $i=1,\dots,n$. Horn functions can be represented by a collection (conjunction) of such implication rules. If for this example we also have $\{x_1,x_3\}\to x_2$ and $\{x_2,x_3\}\to x_1$ as rules of the given Horn system, then we call $\{x_1,x_2\}\to x_3$ circular. The subject of our study is the family of Horn functions that can be represented by a collection of circular implication rules. Since such a circular rule can be represented by the set of the propositional variables appearing in it (since any one is implied by the others), the Horn functions we are interested in can be represented by a hypergraph over the set of variables. For this reason we call this subfamily \textit{hypergraph Horn}. 

The family of hypergraph Horn functions turns out to be highly structured and still very general. They generalize equivalence relations, and more generally matroids. We provide various characterizations for this class, allowing to derive efficient algorithms for several related algorithmic problems. 

We introduce a new type of ``duality'', called implicate-duality or $i$-duality for short. We show that every Boolean function has a unique $i$-dual, which is always hypergraph Horn. Furthermore, $i$-duality is an involution over the set of hypergraph Horn functions. It turns out that $i$-duality generalizes matroid duality. More precisely, if a hypergraph Horn function corresponds to a matroid, then its $i$-dual corresponds to the dual matroid. 
We provide a characterization for self $i$-dual functions, which is surprisingly analogous to the characterization of self dual monotone Boolean functions (see e.g., \cite{bioch1995complexity}).
We also characterize hypergraph Horn functions in terms of the standard closure operator. Our characterization can be regarded as an extension of the Mac Lane\,--\,Steinitz exchange property of matroid closure (see the companion paper \cite{matroid_horn}). 

While satisfiability (SAT) is easy for Horn CNFs, one typical difficulty in this area arises from the fact that Horn functions have many different Horn CNF representations. For instance, an arbitrary Horn function may have several circular rules and several rules that are not circular. It may be representable by a collections of circular rules, but it may also have representations that involve non-circular rules. We provide an efficient algorithm that recognizes if a given Horn CNF defines a hypergraph Horn function or not. If yes, the algorithm outputs a circular representation the size of which is polynomially bounded by the size of the input CNF. We also provide an algorithm that generates all circular rules incrementally efficiently for a Horn function represented by a Horn CNF. Note that there might be an exponential number of such rules.

Keys are another important notion of Horn theory. A subset of the variables is a key of a given Horn function if the truth of these variables imply the truth of all other variables. The set of minimal keys plays important roles in applications, in particular in databases. We provide an efficient algorithm that decides if a given hypergraph is the set of minimal keys of a hypergraph Horn function or not, and if yes, it outputs a circular representation of such a function the size of which is polynomially bounded by the size of the given hypergraph. 
\medskip

The rest of the paper is organized as follows. While we use standard terminology and notation, for clarity we introduce precise definitions in the next section. In Section \ref{s-hypergraph Horn} we introduce the class of hypergraph Horn functions and $i$-duality, and prove a number of structural properties. In Section \ref{sec:char} several equivalent characterizations are shown. In Section \ref{sec-algo1}
we provide the efficient algorithms mentioned above, while in the last section we close with listing a few open problems. 

\section{Preliminaries}
\label{sec:preliminaries}

\paragraph{Hypergraphs.} 
For a finite set $V$, a family $\cH\subseteq 2^V$ of its subsets is called  a {\em hypergraph}, where 
$H \in \cH$ is  called a {\em hyperedge} of $\cH$. A hypergraph $\cH$ is called {\em Sperner} if no distinct hyperedges of $\cH$ contain one another. For a subset $X\subseteq V$ we call $V\setminus X$ its \emph{complement}, and we denote by $\cH^c=\{V\setminus H\mid H\in\cH\}$ the \emph{complementary family} of $\cH$. Note that the operator ``$c$" is an  involution, i.e., $\left(\cH^c\right)^c=\cH$ holds for all hypergraphs. A subset $T\subseteq V$ is called a \emph{transversal} of $\cH$ if $T\cap H\neq\emptyset$ for all hyperedges $H\in\cH$. We denote by $\cH^d$ the \emph{family of minimal transversals} of $\cH$. For simplicity, we omit the parentheses from the notation when applying these operations repeatedly, for example, we write $\cH^{dc}=(\cH^d)^c$ and $\cH^{cd}=(\cH^c)^d$. For a Sperner hypergraph $\cH$, the operator ``$d$'' is also an involution, hence we have $\cH^{dccd}=\cH^{cddc}=\cH$ for such hypergraphs. Note that the family $\cH^{dc}$ is the \emph{family of maximal independent sets} of $\cH$, where a set $I \subseteq V$ is called {\em independent}  of $\cH$ if it contains no hyperedge of $\cH$. The family $\cH^{cd}$ consists of all minimal subsets of $V$ that are not contained in a  hyperedge of $\cH$. Every Sperner hypergraph $\cH$ defines an \emph{independence system} $\cI$ consisting of all independent sets of $\cH$, see e.g.,~\cite{Welsh1976}. In this correspondence,  maximal independent sets of $\cH$ are also called \emph{bases} of the associated independence system $\cI$, while minimal {\em dependent} sets (i.e.,  hyperedges) of $\cH$ are called  \emph{circuits} of $\cI$.
Typically we reserve the notation $\cB$ and $\cC$ for the families of bases and circuits, respectively. Note that for all independence systems the families of bases and circuits are Sperner, and thus we have
\begin{equation*}\label{e-bases-circuits}
	\cB^{cd}=\cC ~~~\text{ and equivalently }~~~ \cC^{dc}=\cB.
\end{equation*}
Independent sets are exactly the subsets of hyperedges of $\cB=\cC^{dc}$, or equivalently, subsets that do not contain a hyperedge of $\cC$. 

For an arbitrary hypergraph $\cH\subseteq 2^V$ we denote by $\cH^\cap$ the \emph{intersection closure} of $\cH$, defined by 
\[
\cH^\cap ~=~ \left\{ \left.\bigcap_{F\in\cF} F ~\right|~ \cF\subseteq \cH \right\}.
\]
We note that the intersection of an empty family is defined as $V$, and thus we have $V\in \cH^\cap$ for all hypergraphs $\cH\subseteq 2^V$. Analogously, we denote by $\cH^\cup$ the \emph{union closure} of $\cH$ defined as
\[
\cH^\cup ~=~ \left\{ \left.\bigcup_{F\in\cF} F ~\right|~ \cF\subseteq \cH \right\}.
\]
We note that the union of an empty family is defined as the empty set, and thus we have $\emptyset\in\cH^\cup$ for all hypergraphs $\cH$. 

\paragraph{Horn functions.} 
We denote by $V$ the set of $n$ Boolean variables $x$ and call these together with their negations $\overline{x}$ as \emph{literals}.  Members of $V$ are called \emph{positive literals} while their negations are \emph{negative literals}. A disjunction of a subset of the literals is called a \emph{clause} if it contains no complementary pair of literals $x$ and $\overline{x}$, and the conjunction of clauses is a called a \emph{conjunctive normal form} (or in short a {\em CNF}). It is well-known (see e.g.,~\cite{CH11}) that every Boolean function $f:\{0,1\}^V\rightarrow\{0,1\}$ can be represented by a CNF (typically not in a unique way). We use Greek letters to denote CNFs, and Latin letters to denote (Boolean) functions. For a CNF $\Psi=\bigwedge_kC_k$, let $|\Psi|$ denote the number of clauses in $\Psi$, and let $\|\Psi\|$ denote the length of $\Psi$, i.e., $\|\Psi\|=\sum_k|C_k|$.

Truth assignments (i.e., Boolean vectors) $\bx=(x_1,\dots,x_n)\in\{0,1\}^V$ can be viewed equivalently as characteristic vectors of subsets. For a subset $Z\subseteq V$ we denote by $\chi_Z\in \{0,1\}^V$ its characteristic vector, i.e., $(\chi_Z)_i=1$ if and only if $i \in Z$. Since we use primarily a combinatorial notation in this paper, we use for a function $f$ and Boolean  expression $\Phi$ the notation $f(Z)$ and $\Phi(Z)$ instead of $f(\chi_Z)$ and $\Phi(\chi_Z)$, to denote the evaluation of $f$ and  $\Phi$ at the binary vector $\chi_Z$, respectively. We say that a set $Z\subseteq V$ is a \emph{true set} of $f$ if $f(Z)=1$, and a \emph{false set} otherwise. We denote by $\cT(f)$ and $\cF(f)$ the \emph{families of true sets} and \emph{false sets} of $f$, respectively. Every Boolean expression defines/represents a unique Boolean function. If $A$ and $B$ denote Boolean functions or Boolean expressions, we write $A=B$ if $A(X)=B(X)$ for all $X\subseteq V$, and  $A\leq B$ if $A(X)\leq B(X)$ for all $X\subseteq V$, where $B$ is called a \emph{majorant} of $A$ in the latter case. We also write $A <B$ if $A\leq B$ and $A\not=B$. 

A clause is called \emph{definite Horn} if it contains exactly one positive literal. 
It is easy to see that definite Horn clauses represent simple implications. Namely, for a proper subset $B\subsetneq V$ and $v\in V\setminus B$  the implication $B\to v$ is equivalent to the definite Horn clause $C=v\vee \left(\bigvee_{u\in B}\overline{u}\right)$: the true sets of both expressions are exactly the sets $T\subseteq V$ such that either $T\not\supseteq B$ or $T\supseteq  B\cup\{v\}$. 
A CNF is called  {\em definite Horn} if it consists of definite Horn clauses, and a Boolean function is called  {\em definite Horn} if it can be represented by a definite Horn CNF. 
The following characterization of definite Horn functions is  well-known. 

\begin{lemma}[see e.g.,
\cite{McKinsey1943,CH11}]
\label{l-E-intersection-closed}
	A Boolean function $f$ is definite Horn if and only if the family $\cT(f)$  of its true sets is closed under intersection and contains $V$.
	\hfill$\Box$
\end{lemma}

Lemma~\ref{l-E-intersection-closed} implies that for any  set $Z\subseteq V$ there exists a unique minimal true set containing $Z$, the so-called \emph{closure} of $Z$ denoted by $\FC_h(Z)$. In fact such a closure can be computed efficiently from  any definite Horn CNF representation $\Phi$ of $h$ by the so-called \emph{forward chaining procedure} (see e.g.,~\cite{CH11}): Let $A\subseteq V$ denote the set of all variables $v\in V \setminus Z$ for which there exists a clause $B\to v$ in $\Phi$ with $B\subseteq Z$ and $v \in V \setminus Z$, and define $\FC^1_{\Phi}(Z)\coloneqq Z\cup A$. For $i\geq 2$ we define $\FC^i_{\Phi}(Z)=\FC^1_{\Phi}(\FC^{i-1}_{\Phi}(Z))$. Since $ Z\subseteq \FC^{1}_{\Phi}(Z) \subseteq  \FC^{2}_{\Phi}(Z) \subseteq \dots \subseteq V$, $\FC^{i+1}_{\Phi}(Z)=\FC^i_{\Phi}(Z)$ holds for some integer $i\leq n$. Let $i^*$ be the smallest such index $i$. Then we have $\FC^t_{\Phi}(Z)=\FC^{i^*}_{\Phi}(Z)$ for all $t\geq i^*$, and $\FC^{i^*}_{\Phi}(Z)$ is the minimal true set of $\Phi$ that contains $Z$. Thus we can define $\FC_{\Phi}(Z)=\FC^{i^*}_{\Phi}(Z)$, and we say that  {\em $Z$ can be closed by $\Phi$ in $i^*$ steps}. While we may have $\FC^1_{\Phi}(Z)\neq \FC^1_{\Psi}(Z)$ for different definite Horn CNFs $\Phi$ and $\Psi$ representing the same function $h$, it can be shown (see e.g.,~\cite{CH11}) that the resulting set $\FC_\Phi(Z)$ does not depend on the particular choice of the representation $\Phi$ of $h$, but only on the underlying function $h$, that is, we can write $\FC_h(Z)=\FC_\Phi(Z)$ for this uniquely defined closure of $Z$.  Note that $\FC$ is in fact a closure operator in finite set theory,  and hence we call a subset $Z\subseteq V$ \emph{closed} (with respect to $h$) if $\FC_h(Z)=Z$. It is not difficult to check that a set is closed with respect to $h$ if and only if it is a true set of $h$. 

For a definite Horn function $h$, a clause $B\to v$ is called an \emph{implicate} of $h$ if it is a majorant of $h$, that is,  if $(B\to v) \geq h$. An implicate $B\to v$ of $h$ is \emph{prime} if $h$ has no other implicate $B'\rightarrow v$ with $B'\subsetneq B$. The following lemma characterizes  implicates of definite Horn functions in terms of the closure operator.  

\begin{lemma}[see e.g.,~\cite{CH11}]\label{l-E-implicate}
A clause $B\to v$ is an implicate of a definite Horn function $h$ if and only if $v\in \FC_h(B)\setminus B$. 
\hfill$\Box$
\end{lemma}

Given a definite Horn CNF $\Phi$, a subset $A\subseteq V$ and a variable $v\in V\setminus A$, we write $A\overset{\Phi}{\to}v$ to indicate that $A\to v$ is an implicate of $\Phi$, or equivalently that $v\in \FC_{\Phi}(A)$. To indicate the opposite, that is that $A\to v$ is not an implicate of $\Phi$, or equivalently, that $v\not\in \FC_{\Phi}(A)$ we may simply write $A\overset{\Phi}{\nrightarrow}v$.

A subset $K\subseteq V$ is called a \emph{key} of the definite Horn function $h$ if $\FC_h(K)=V$. We denote by 
\begin{equation*}
	\cK(h)=\{K\subseteq V\mid \FC_h(K)=V\ \text{and}\ \FC_h(K')\neq V\ \text{for all $K'\subsetneq K$}\},
\end{equation*}
the family of \emph{minimal keys} of $h$, and $\cK(h)$ is called  the {\em key set} of $h$.

A true set $T\in\cT(h)$ is called \emph{nontrivial} if $T\neq V$. We denote by 
\begin{equation*}
	\cM(h)=\{T\subsetneq V\mid h(T)=1\ \text{and}\ h(T')=0\ \text{for all $T\subsetneq T'\subsetneq V$}\},
\end{equation*}
the family of \emph{maximal nontrivial true sets} of $h$. 
Note that $\cM(h)$ is a subfamily of the so-called \emph{characteristic models} of $h$~\cite{KhardonRoth1996}.

Both families, $\cK(h)$ and $\cM(h)$ are Sperner hypergraphs for all definite Horn functions $h$. 
It is easy to verify that maximal nontrivial true sets form the family of maximal independent sets of the family of minimal keys, and minimal keys are exactly the minimal sets that are not contained in a maximal nontrivial true set.

\begin{lemma}[see e.g.,~\cite{CH11}]\label{l-E-M(h)-K(h)}
For a definite Horn function $h$ we have $\cM(h)=\cK(h)^{dc}$ and $\cK(h)=\cM(h)^{cd}$.
\hfill$\Box$
\end{lemma}

For a Sperner hypergraph $\cH$ let us denote by $\hmk(\cH)$ and $\hmt(\cH)$ the sets of definite Horn functions $h$ such that $\cK(h)=\cH$ and $\cM(h)=\cH$, respectively. The following lemma will be used in the subsequent sections.

\begin{lemma}\label{l-minkey-maxtrue}
Given a hypergraph $\cH\subseteq 2^V$, the following hold.
\begin{enumerate}[label={\rm (\alph*)}]\itemsep0em
	\item \label{maxtrue:i} $\hmk(\cH)=\hmt(\cH^{dc})$.
	\item \label{maxtrue:ii} If $f,g\in \hmk(\cH)$ then $f\wedge g\in\hmk(\cH)$.
	\item \label{maxtrue:iii} If $f,g\in\hmt(\cH)$ then $f\wedge g\in\hmt(\cH)$.
\end{enumerate}
\end{lemma}
\begin{proof}
Lemma~\ref{l-E-M(h)-K(h)} implies \ref{maxtrue:i} and the equivalence of \ref{maxtrue:ii} and \ref{maxtrue:iii}. We prove here \ref{maxtrue:iii}. Note that $X$ is a true set of $f\wedge g$ if and only if it is a true set of both $f$ and $g$. Thus, if $\cM(f)=\cM(g)=\cH$, then we must have $\cM(f\wedge g)=\cH$, too. 
\end{proof}

\paragraph{Matroids.}

Matroids were introduced by Whitney~\cite{whitney1992abstract} and independently by Nakasawa~\cite{nishimura2009lost} as abstract generalizations of linear independence in vector spaces. Matroids are special independence systems satisfying stronger properties. There are several equivalent, well-known axiomatizations, see e.g.,~\cite{Welsh1976,Tutte1971}. We recall here one of those that will be used in our analysis. A nontrivial Sperner hypergraph $\cC\subseteq 2^V$ is the \emph{family of circuits} of a matroid if and only if
\begin{equation}\leqnomode
\forall C_1,C_2\in\cC,\ C_1\neq C_2 \text{ and } u\in C_1\cap C_2 ~~\exists~ C_3\in\cC 
\text{ s.t. } C_3\subseteq (C_1\cup C_2) - u.
\tag{\textbf{C}}\label{CC}
\end{equation}

\section{Hypergraph Horn functions}
\label{s-hypergraph Horn}

To simplify our notation, we respectively write  $H+v$ and $H-v$ instead of $H\cup\{v\}$ and $H\setminus\{v\}$ for $H\subseteq V$ and $v \in V$. Given a definite Horn function $h:2^V\to \{0,1\}$, we call an implicate $A\to v$ of it \emph{circular} if $((A+v)-u)\to u$ is also an implicate of $h$ for every $u\in A$. To a hypergraph $\cH\subseteq 2^V$ we associate the definite Horn CNF $\Phi_\cH$ defined as
\begin{equation*}\label{e-circulant}
	\Phi_\cH ~=~ \bigwedge_{H\in\cH} \left(\bigwedge_{v\in H} \left((H-v)\to v\right)\right)
\end{equation*}
and call $\Phi_{\cH}$ the \emph{circular CNF} associated to the hypergraph $\cH$. Note that circular CNFs generalize equivalence relations. Namely, if all hyperedges of $\cH$ are of size two, then $\Phi_{\cH}$ represents an equivalence relation, and all equivalence relations can be represented in this way. For example, if $\cH=\{\{1,2\}, \{2,3\}, \{4,5\}\}$, then $\Phi_\cH$ represents two equivalence classes: $\{1,2,3\}$ and $\{4,5\}$. We remark that even equivalence relations can be represented by different hypergraphs. For example, $\cH=\{\{1,2\},\{2,3\}, \{4,5\}\}$ and $\cG=\{\{1,2\},\{1,3\}, \{4,5\}\}$ satisfy $\Phi_\cH=\Phi_{\cG}$.

We say that a definite Horn function $h:2^V\to \{0,1\}$ is \emph{hypergraph Horn} if it has a circular CNF representation, that is, if there exists a hypergraph $\cH\subseteq 2^V$ such that $h ~=~ \Phi_\cH$. The next example shows that $\Phi_\cH=\Phi_{\cH'}$ can hold even if $\cH$ is Sperner while $\cH'$ is not, and that hypergraph Horn functions might have non-circular implicates.

\begin{ex}\label{ex-1}
Consider the set $V=\{1,2,\dots,8\}$, the subsets $H_1=\{1,2,5\}$, $H_2=\{4,5,6\}$, $H_3=\{3,4,7\}$, $H_4=\{2,7,8\}$, $H_5=\{1,2,3,4\}$, $H_6=\{1,2,3,5,6\}$, and $H_7=\{1,3,4,7,8\}$, and the hypergraphs
\begin{align*}
	\cH &=~ \{H_1,H_2,H_3,H_4,H_5\}, \text{ and}\\
	\cH' &=~ \{H_1,H_2,H_3,H_4,H_6,H_7\}.
\end{align*}
We claim that $\Phi_{\cH}=\Phi_{\cH'}$. To see this, we show that $\Phi_{\{H_5\}}\geq \Phi_{\cH'}$ and that $\Phi_{\{H_6,H_7\}}\geq \Phi_{\cH}$. These relations then imply $\Phi_{\cH}\leq \Phi_{\cH'}$ and $\Phi_{\cH'}\leq \Phi_{\cH}$ from which $\Phi_{\cH}=\Phi_{\cH'}$ follows. 

To see e.g., that $\Phi_{\{H_5\}}$ is a majorant of $\Phi_{\cH'}$ one can show that $\FC_{\Phi_{\cH'}}(H_5-u)\ni u$ for all $u\in H_5$. For instance, for $u=1$ and $S=H_5-u=\{2,3,4\}$ we can see that $\FC^1_{\cH'}(S)=S+7$ because of the clause $H_3-7\to 7$ of $\cH'$. Then we get $\FC^1_{\cH'}(S+7)=S+7+8$ because of the clause $H_4-8\to 8$ of $\cH'$. Finally, we get $\FC_{\cH'}^1(S+7+8)=S+7+8+1$ due to the clause $H_7-1\to 1$. Thus, we have $1\in \FC_{\cH'}(\{2,3,4\})$. By Lemma~\ref{l-E-implicate}, this implies that $H_5-1\to 1$ is an implicate of $\Phi_{\cH'}$. Using similar arguments for the other elements of $H_5$, we can show that $H_5-2\to 2$, $H_5-3\to3$ and $H_5-4\to 4$ are all implicates of $\Phi_{\cH'}$, implying the inequality $\Phi_{\cH}\geq \Phi_{\cH'}$. Applying similar arguments to $H_6$ and $H_7$, we can also prove the reverse inequality, as claimed above.

Note that in this example the hypergraph $\cH$ is Sperner, while $\cH'$ is not. Furthermore, in this example there exists a non-circular implicate of $\Phi_\cH$. Namely,  $\{1,5,6\}\to 3$ is an implicate of $\Phi_\cH$ while $\{3,5,6\}\to 1$ is not. 
\end{ex}

We say that a subset $I\subseteq V$ is an \emph{implicate set} of a definite Horn function $h$ if $(I-v) \to v$ is an implicate of $h$ for all $v\in I$. For a Boolean function $f$, we denote by $\cI(f)$ its \emph{family of implicate sets}.  By definition, we have $\emptyset\in\cI(f)$ for all Boolean functions $f$. 

\begin{lemma}\label{l-implicate-sets-are-union-closed}
For every Boolean function $f:2^V\to \{0,1\}$ the hypergraph $\cI(f)$ is union closed, that is,
$I,J\in\cI(f)$ implies $I\cup J\in\cI(f)$.
\end{lemma}
\begin{proof}
Take an arbitrary element $v\in I\cup J$. Without loss of generality, we may assume that $v\in I$. As $I$ is an implicate set of $f$, we have $(I-v)\to v$. This implies that $((I\cup J)-v)\to v$ holds.
\end{proof}

Given two families $\cG\subseteq \cF\subseteq 2^V$, we call $\cG$ a \emph{generator} of $\cF$ if its union closure is $\cF$, that is, $\cG^\cup=\cF$.  For every Boolean function $f:2^V\to\{0,1\}$, there exists a unique minimal generator of $\cI(f)$. Indeed, let $\cG(f)$ consist of those implicate sets that are not the union of others. These sets are clearly contained in any subfamily whose union closure is $\cI(f)$, and every other set in $\cI(f)$ can be obtained as the union of such sets. We call $\cG(f)$ the \emph{standard generator} of $\cI(f)$. Note that $\cG(f)$ may not form a Sperner hypergraph.

\begin{ex}\label{ex-is-beyond-union}
Let us note that the set $\cI(h)$ can be much larger than the input representation. Furthermore, it can happen that $\cI(h)$ cannot be obtained by taking unions of the input, even for hypergraph Horn functions. 

For a set $V=\{i,i'\mid i=1,\dots ,n\}$ of $2n$ elements, consider the hypergraph $\cH=\{  \{i,i'\} \mid i=1,\dots ,n \} \cup \{ \{1,2, \dots ,n\}\}$ of $n+1$ hyperedges and the corresponding hypergraph Horn function $h=\Phi_{\cH}$. It is not difficult to see that 
\[
\cI(h)=(\cH \cup \left\{Z\subseteq V\mid Z\cap \{i,i'\}\neq \emptyset ~\forall~ i=1,\dots,n\right\})^{\cup}. 
\]
This family of implicate sets contains exponentially many hyperedges that are not in the union closure of the input. 
\end{ex}

As pointed out in Lemma~\ref{l-implicate-sets-are-union-closed}, implicate sets of a Boolean function $f$ are closed under taking unions. This implies that every subset $S\subseteq V$ has a unique maximal implicate set as its subset. We denote by $\IS_f(S)$ this unique maximal implicate subset of $S$ and call it the \emph{core} of $S$ (or the $f$-\emph{core} of $S$, if the function is not clear from the context). Recall that $\cT(f)$ denotes the family of true sets of $f$. 

\begin{lemma}\label{c-E-truesets-implicatesets}
For a definite Horn function $f$, we have the following equivalences:
\begin{subequations}
\begin{equation}\label{e-truesets=FC}
T\in\cT(f) \,\,\iff\,\, \FC_f(T)=T,
\end{equation}
\begin{equation}\label{e-implicate-set-of-haa}
I\in \cI(f) \,\, \iff\,\, \IS_f(I)=I.
\end{equation}
\end{subequations}
\end{lemma}
\begin{proof}
The lemma follows by the definitions, as for a subset $S\subseteq V$ the closure $\FC_f(S)$ is the unique smallest true set containing $S$, while the core $\IS_f(S)$ is the unique maximal implicate set contained in $S$.
\end{proof}

The next technical lemma, together with its corollary,  describes a certain duality relation between true sets and implicate sets specifically for hypergraph Horn functions, which  will be useful in our proofs later. 

\begin{lemma}\label{l-E-true-set-implicate-set-basics}
For a hypergraph Horn function $h$ and subsets $I,T\subseteq V$,  we have the following equivalences:
\begin{subequations}
\begin{equation}\label{e-true-set-of-h}
T\in\cT(h) \,\,\iff \,\,\nexists \, J\in\cI(h) \text{ with } ~ |J\setminus T|=1,
\end{equation}
\begin{equation}\label{e-implicate-set-of-h}
I\in \cI(h)  \,\,\iff\,\, \nexists \,  S\in\cT(h) \text{ with } ~ |I\setminus S|=1. 
\end{equation}
\end{subequations}
\end{lemma}
\begin{proof}
For \eqref{e-true-set-of-h}, note that the right hand side is equivalent with $\FC^1_{\Phi_{\cI(h)}}(T)=T$, which is indeed equivalent with $T\in\cT(h)$ since $h=\Phi_{\cI(h)}$ is hypergraph Horn.

For \eqref{e-implicate-set-of-h},  note that the right hand side is equivalent to the claim that for every $u\in I$ and for every $S\in\cT(h)$ the inclusion $I-u\subseteq S$ implies $u\in S$. This is  equivalent to $\FC_h(I-u)\ni u$ for all $u\in I$, which is the definition of $I$ being an implicate set of $h$. 
\end{proof}

As a direct consequence, we get the following characterizations of true and implicate sets for hypergraph Horn functions.

\begin{corollary}\label{c-imp-true}
For a hypergraph Horn function $h$,  we have the following equalities:
\begin{subequations}
\begin{equation}\label{e2-true-set-of-h}
\cT(h) ~=~ \{ T\subseteq V \mid    \nexists ~ I\in\cI(h) \text{ with } ~ |I\setminus T|=1  \},
\end{equation}
\begin{equation}\label{e2-implicate-set-of-h}
\cI(h) ~=~ \{ I\subseteq V \mid \nexists ~ T\in\cT(h) \text{ with } ~ |I\setminus T|=1 \}.
\end{equation}
\end{subequations}
\hfill$\Box$
\end{corollary}

Let us remark that \eqref{e-implicate-set-of-h} and \eqref{e2-implicate-set-of-h} hold for arbitrary Boolean functions $h$, not only for hypergraph Horn ones. Moreover, later we show that both \eqref{e-true-set-of-h} and  \eqref{e2-true-set-of-h} provide necessary and sufficient conditions for $h$ being hypergraph Horn. 

Observe that the conjunction of two hypergraph Horn functions $h_1$ and $h_2$ is also hypergraph Horn. Indeed, if $\cH_1$ and $\cH_2$ are hypergraph such that $h_i=\Phi_{\cH_i}$ for $i=1,2$, then $h_1\wedge h_2=\Phi_{\cH_1\cup\cH_2}$. Therefore any Boolean function $f$ admits a unique minimal hypergraph Horn majorant which we denote by $f^{\circ}=\Phi_{\cI(f)}$. Note that $f^{\circ}=1$ may hold for a Boolean function $f$. Note also that $f$ is a hypergraph Horn function if and only if $f=f^{\circ}$. 

The natural ``duality'' between the families of true sets and implicate sets pointed out above motivates the following definition. To a given Boolean function $f$ we associate another Boolean function, called its \emph{implicate-dual} and denoted by $f^i$, defined uniquely by the equality
\begin{equation}\label{e-f->f^i}
\cT(f^i) ~=~ \cI(f)^c.
\end{equation}
Since $\cI(f)=\cI(f^\circ)$ for every Boolean function $f$, we have $f^i=(f^\circ)^i=f^{\circ i}$. Furthermore, since $\emptyset\in \cI(f)$ we have $V\in\cT(f^i)$ and by Lemma~\ref{l-implicate-sets-are-union-closed} the family $\cI(f)^c$ is closed under intersections. Consequently, by Lemma~\ref{l-E-intersection-closed}, the function $f^i=f^{\circ i}$ is definite Horn for every Boolean function $f$.

\section{Characterizations of hypergraph Horn functions}
\label{sec:char}

In this section we provide multiple characterizations of hypergraph Horn functions. Recall that for a Boolean function $f$, $\cG(f)$ denotes the standard generator of $\cI(f)$. Our main theorem below summarizes these equivalences.

\begin{theorem}\label{t-MAIN-hH}
For a definite Horn function $f:2^V\to \{0,1\}$, the following claims are equivalent.
\begin{enumerate}[label={\rm (\roman*)}]\itemsep0em
\item \label{it:a} $f$ is hypergraph Horn.
\item \label{it:b} $f^\circ=f$.
\item \label{it:c} $f^{ii}=f$.
\item \label{it:d} For every false set  $F$ of $f$, there exists an implicate set $I\in\cI(f)$ such that $|I\setminus F|=1$. 
\item \label{it:e} For every false set  $F$ of $f$, there exists an implicate set $I\in\cG(f)$ such that $|I\setminus F|=1$.
\item \label{it:f} For every false set $F$ of $f$, there exists a variable $u\in \FC_f(F)\setminus F$ such that $v\in \FC_f(F-v+u)$ holds for every $v\in F$ with $f(F-v)=1$. 
\end{enumerate}
\end{theorem}

Let us note that property \ref{it:c} is a generalization of \emph{matroid duality}. Property \ref{it:d} above is the same as  \eqref{e-true-set-of-h} and  \eqref{e2-true-set-of-h} and characterizes hypergraph Horn functions. This property will be used in Section~\ref{sec-recog} for an efficient algorithm to recognize hypergraph Horn functions from definite Horn CNFs. Property \ref{it:f} is a generalization of the so-called {\em Mac Lane\,--\,Steinitz exchange property} for closure operator of matroids~\cite{Welsh1976}.

In what follows, we show several simpler claims and we use those put together at the end of this section to prove Theorem~\ref{t-MAIN-hH}. Our first result shows the equivalence of \ref{it:b} and \ref{it:d} in Theorem~\ref{t-MAIN-hH}.

\begin{lemma}\label{t-(b)<=>(c)}
For a Boolean function $f:2^V\to\{0,1\}$, we have $f^\circ=f$ if and only if for every false set $F$ of $f$ there exists an implicate set $I\in \cI(f)$ such that $|I\setminus F|=1$.
\end{lemma}
\begin{proof}
Since $f=f^\circ$ is equivalent to the condition that $f$ is hypergraph Horn, the only-if part follows from \eqref{e-true-set-of-h}. For the if part, 
assume that $f\neq f^\circ$. Since $f\leq f^\circ$ holds for all Boolean functions, there exists a subset $F\subseteq V$ for which $f(F)=0$ and $f^\circ(F)=1$. By applying \eqref{e-true-set-of-h} to $f^\circ$, no implicate set $I\in \cI(f)$ satisfies  $|I\setminus F|=1$, which completes the proof of the if part. 
\end{proof}

Next we provide a characterization of hypergraph Horn functions that proves the equivalence of \ref{it:a} and \ref{it:f} of Theorem~\ref{t-MAIN-hH}. 

\begin{lemma}\label{t-hH-characterization}
A definite Horn function $h$ is hypergraph Horn if and only if for every false set $F$ of $h$,  there exists a variable $u\in V\setminus F$ such that
\begin{enumerate}[label={\rm (\roman*)}]\itemsep0em
\item \label{char:i} $F\overset{h}{\to} u$, and 
\item \label{char:ii} $F-v+u\overset{h}{\to}v$ for all variables $v\in F$ with $h(F-v)=1$. 
\end{enumerate}
\end{lemma}
\begin{proof}
Consider a hypergraph Horn function $h$, i.e., $h=\Phi_{\cI(h)}$. Let us consider an arbitrary false set $F\subseteq V$ of $h$. Then  \eqref{e-true-set-of-h} implies that there exists an implicate set $I\in\cI(h)$ with  $I\setminus F=\{u\}$ for some $u\in V$. We claim that $u$ satisfies \ref{char:i} and \ref{char:ii} in Lemma~\ref{t-hH-characterization}. Indeed,  $I\setminus F=\{u\}$ implies $F\overset{h}{\to} u$. Moreover, for any $v \in F$ with $h(F-v)=1$, we have $I\setminus (F-v)=\{v,u\}$, as otherwise  $I\setminus (F-v)=\{u\}$ holds, implying that $F-v$ is a false set of $h$, a contradiction. Therefore, we have $(F-v+u)\overset{h}{\to}v$, which proves \ref{char:ii}. 

For the reverse implication, assume that $h < h^\circ$ (i.e., $h \not=h^\circ$), and choose a minimal subset $F\subseteq V$ such that $h(F)=0$ and $h^\circ(F)=1$. By the minimality of $F$ and by $h\leq h^\circ$, for every variable $v\in F$ we have either $h(F-v)=1$ or $h^\circ(F-v)=0$. Let us define $Q=\{v\in F\mid h(F-v)=1\}$ and $R= F\setminus Q=\{v\in F\mid h^\circ(F-v)=0\}$. Since $h^\circ(F)=1$ and  $h^\circ(F-v)=0$ for a variable $v\in R$,  there exists an implicate set $I\in\cI(h)=\cI(h^\circ)$ such that $v\in I\subseteq F$. This implies that $(F-v)\overset{h}{\to}v$ for all $v\in R$. If we assume that some $u \in V \setminus F$ satisfies \ref{char:i} and \ref{char:ii}, then  $F+u$ is an implicate set of $h$. This however contradicts $h^\circ(F)=1$, which completes the proof of the if part of the theorem.
\end{proof}

We say that a hypergraph $\cH$ is a \emph{minimal representation} of a hypergraph Horn function $h$ if $\Phi_\cH=h$ but $\Phi_{\cH'}\neq h$ for every $\cH'\subsetneq\cH$. The next lemma, that is of independent combinatorial interest on its own, shows that hypergraph Horn functions can be represented by their standard generators; this in turn implies the equivalence of \ref{it:d} and \ref{it:e}.

\begin{lemma} \label{lem:minrep}
If $\cH\subseteq\cI(h)$ is a minimal representation of the hypergraph Horn function $h$ and $H\in\cH\setminus\cG(h)$, then $H$ can be replaced by a member $H'$ of $\cG(h)$ such that for the resulting hypergraph $\cH'=\cH\setminus \{H\}\cup\{H'\}$ we have $\Phi_{\cH'}=h$.
\end{lemma}
\begin{proof}
Take an arbitrary hyperedge $H\in\cH\setminus\cG(h)$, and let $\cJ=\{J_1,\dots,J_k\}$ be an inclusionwise minimal family of sets from $\cG(h)$ whose union is $H$. Let $H_j= \bigcup_{\ell\neq j} J_\ell$. By the minimality of $\cJ$, $H_j$ is a proper subset of $H$ for every $j$. 

We claim that $H$ can be replaced by $H_j$ for some $j$; the repeated application of such a step immediately proves the theorem. Suppose to the contrary that the claim does not hold. Then for each $j$, by Lemma~\ref{l-E-true-set-implicate-set-basics} and the minimality of $\cH$, there exists an element $x_j\in H\setminus H_j$ and a false set $F_j\in\cF(h)$ such that $H\setminus F_j=\{x_j\}$ and $H$ is the unique hyperedge of $\cH$ with that property. Let $F= F_1\cap F_2$. As $|J_1\setminus F|=1$, $F$ is a false set by Lemma~\ref{l-E-true-set-implicate-set-basics}. Since $H$ is the only set in $\cH$ such that $|H\setminus F_j|=1$ for both $j=1$ and $2$, there exists no hyperedge $K\in\cH$ with $|K\setminus F|=1$. This contradicts our assumption that $\Phi_\cH=h$. 
\end{proof}

The equivalence of \ref{it:d} and \ref{it:e} then follows: \ref{it:d} is clearly a strengthening of \ref{it:e}, while the other direction holds as $h=\Phi_{\cG(h)}$ by Lemma~\ref{lem:minrep}. It is worth formulating the latter observation as a corollary.

\begin{corollary}
Any hypergraph Horn function $h$ admits a minimal representation $\cH$ such that $\cH\subseteq\cG(h)$. \hfill$\Box$
\end{corollary}

We show next that implicate-duality associates a hypergraph Horn function to every Boolean function. 

\begin{lemma}\label{t-f^i-is-hH}
The function $f^i$ is hypergraph Horn for every Boolean function $f$.
\end{lemma}
\begin{proof}
By the definition of the $i$ operator in \eqref{e-f->f^i}, $f^i$ is always definite Horn. Thus, it suffices to show that every false set $F$ of $f^i$ has a variable $v \in V\setminus F$ that satisfies \ref{char:i} and \ref{char:ii} of Lemma ~\ref{t-hH-characterization}, since this proves that $f^i$ is hypergraph Horn by Lemma~\ref{t-hH-characterization}. 

Let $U =\FC_{f^i}(F) \setminus F$. Since $F$ is a false set of $f^i$, we have $U \not=\emptyset$ and  
\begin{equation*}\label{e-F->u}
F\overset{f^i}{\to}u ~~~\mbox{ for all}~ u\in U,
\end{equation*}
implying that \ref{char:i} of Theorem~\ref{t-hH-characterization} holds for any $u\in U$. Let us denote by 
\[
W ~=~ \{ w\in F\mid F-w\in\cT(f^i)\}.
\]
If $W=\emptyset$, then \ref{char:ii} of Theorem~\ref{t-hH-characterization} is vacuous for this false set $F$, and we have nothing else to show. Thus, let us assume for the rest of the proof that $W\neq \emptyset$. Let $\bar{F}= V\setminus F$,  $I_0= \bar{F}\setminus U$, and $I_w= \bar{F}+w$ for $w \in W$. Then, by the definition of the $i$ operator in \eqref{e-f->f^i}, we have $\bar{F}\notin \cI(f)$, $I_0 \in \IS_f(\bar{F})$, and $I_w \in\cI(f)$  for every $w \in W$. Let us observe that $\bar{F}-u\overset{f}{\to}u$ holds for all $u\in I_0$, since $\bar{F}-u\supseteq I_0-u$ and $I_0$ is an implicate set of $f$. Therefore, by $\bar{F}\notin \cI(f)$, there exists a variable $u\in U$ such that
\begin{equation}
\label{eq-aege1}
\bar{F}-u\overset{f}{\nrightarrow} u.
\end{equation}
We claim that this $u \in U$ satisfies \ref{char:ii} of Lemma~\ref{t-hH-characterization}, which complete the proof of the lemma.  

To prove the claim, let us assume that there exists a variable $w\in W$ such that $F-w+u\overset{f^i}{\nrightarrow} w$. Let $I_{uw}$ be an implicate set  of $f$ defined by $I_{uw}=V\setminus \FC_{f^i}(F-w+u)$, where we note that  $\FC_{f^i}(F-w+u)$ is a true set of $f_i$. Since $F-w+u\overset{f^i}{\nrightarrow} w$, we have $w \not\in \FC_{f^i}(F-w+u)$, i.e., $w\in I_{uw}$. 
Thus, $I_{uw}\setminus (\bar{F}-u)=\{w\}$ holds. This, together with $I_w= \bar{F}+w$, implies  that $\FC_f(\bar{F}-u) \supseteq \bar{F}+w$. This means by Lemma~\ref{l-E-implicate} that $\bar{F}-u\to u$ is an implicate of $f$, which contradicts \eqref{eq-aege1}. 
\end{proof}

Next we show that implicate-duality is an involution when restricted to the family of hypergraph Horn functions. More precisely, we prove the following statement. 

\begin{lemma}\label{t-ii=o}
We have $f^{ii}=f^\circ$  for every Boolean function $f$. 
\end{lemma}
\begin{proof}
Since a Boolean function is uniquely determined by its family of true sets, the statement is equivalent to the equality
\begin{equation}\label{e-Ifi=Tfcirc}
\cT(f^\circ)  ~=~ \cI(f^i)^c. 
\end{equation}
by the definition of the implicate-dual of $f^i$. Note that $f^\circ$ and $f^i$ are both  hypergraph Horn by the definition of $f^\circ$ and Lemma~\ref{t-f^i-is-hH}. Thus, we shall show \eqref{e-Ifi=Tfcirc} by applying Lemma~\ref{l-E-true-set-implicate-set-basics} for both $f^\circ$ and $f^i$. By applying \eqref{e-true-set-of-h} of Lemma~\ref{l-E-true-set-implicate-set-basics} to $f^\circ$,  we get the equivalence
\[
T\in \cT(f^\circ) 
\,\,\iff \,\,
\nexists \,J\in\cI(f^\circ) \text{ s.t. } |J\setminus T|=1. 
\]
Since the true sets of $f^i$ are exactly the complements of implicate sets of $f^\circ$ by our definition of the implicate-dual, we also have (for an arbitrary set $T$) the equivalence 
\[
\nexists J\in\cI(f^\circ) \text{ s.t. } |J\setminus T|=1 
\,\,\iff\,\,
\nexists S\in \cT(f^i) \text{ s.t. } |\bar{T}\setminus S|=1. 
\]
Finally, by applying \eqref{e-implicate-set-of-h} of Lemma~\ref{l-E-true-set-implicate-set-basics} to $f^i$, we get the equivalence of
\[
\nexists S\in \cT(f^i) \text{ s.t. } |\bar{T}\setminus S|=1
\iff
\bar{T}\in \cI(f^i).
\]
Putting together the above three equivalences, we obtain \eqref{e-Ifi=Tfcirc}, and the lemma follows.
\end{proof}

Now we are ready to prove the main theorem of this section.

\begin{proof}[\textbf{Proof of Theorem~\ref{t-MAIN-hH}}:]
Note first that the equivalence of \ref{it:a} and \ref{it:b} follows by our definitions. The equivalence of \ref{it:b} and \ref{it:d} follows by Lemma~\ref{t-(b)<=>(c)}. The equivalence of \ref{it:d} and \ref{it:e} follows by Lemma~\ref{lem:minrep}. The equivalence of \ref{it:a} and \ref{it:f} is implied by Lemma~\ref{t-hH-characterization}. Finally, the equivalence of \ref{it:b} and \ref{it:c} follows by Lemma~\ref{t-ii=o}. 
\end{proof}

\begin{remark}
Let us note that the equivalence between \ref{it:a}, \ref{it:b}, \ref{it:c}, \ref{it:d}, and \ref{it:e} in Theorem~\ref{t-MAIN-hH} follows for an arbitrary Boolean function $f$. Property \ref{it:f} makes sense only for definite Horn functions, since the closure operator $\FC$ is defined only when $f$ is definite Horn.  
\end{remark}

We also remark that Lemmas~\ref{t-f^i-is-hH} and~\ref{t-ii=o} imply the following property on hypergraph Horn functions. 

\begin{corollary}\label{c-E-hH-is-idual}
Operator $i$ provides a bijection on hypergraph Horn functions, and  hence all hypergraph Horn functions arise as implicate-duals of hypergraph Horn functions. 
\end{corollary}
\begin{proof}
By Theorem~\ref{t-MAIN-hH}, we have  $h^{ii}=h$ for every hypergraph Horn function $h$. This implies the statement. \end{proof}

Corollary~\ref{c-E-hH-is-idual} naturally raises the following question: how to characterize dual pairs $h,h^i$ of hypergraph Horn functions in terms of their implicate sets, for example? In Section~\ref{sec-algo1}, we discuss the related topics from an algorithmic point of view.   

Let us call a hypergraph Horn function $h$  \emph{self implicate-dual} if $h^i=h$. Even the characterization of hypergraphs $\cH\subseteq 2^V$ for which $h=\Phi_\cH$ is self implicate-dual is open. We shall see in an accompanying paper~\cite{matroid_horn} that the families of circuits of self-dual matroids are such hypergraphs. However, there are self implicate-dual functions that do not correspond to matroids, as shown by the following example.

\begin{ex}\label{ex-fi=f}
For $V=\{1,2,3,4,5\}$, consider a hypergraph Horn function $f$ represented by $\Phi_\cH$, where 
\[
\cH=\left\{ \emptyset, \{1,2,3\}, \{4,5\}, \{1,2,4,5\}, \{1,3,4,5\}, \{2,3,4,5\}, V \right\}. 
\]
It is not difficult  to verify that \[
\cI(h)=\cH ~~~\text{ and }~~~ \cT(h)=\cH^c.
\]
implying that  $h^i=h$, i.e., $h$ is self implicate-dual. On the other hand, minimal nontrival sets in $\cH$ do not satisfy the circuit axiom \eqref{CC}, implying that $f$ is not matorid-Horn; for further details, we refer the interested reader to~\cite{matroid_horn}.
\end{ex}

\section{Computational problems for hypergraph Horn functions}
\label{sec-algo1}

In this section, we consider computational problems for hypergraph Horn functions. We first show that the core $\IS_\Psi(S)$ can be computed in polynomial time from  a  definite Horn CNF $\Psi$ and an arbitrary set $S\subseteq V$. Note that operators $\FC$ and $\IS$ play a dual role in hypergraph Horn functions. While $\FC_\Psi(S)$ is possible to compute in linear time,  $\IS_\Psi(S)$ seems to be difficult to compute in linear time. We then show that the recognition problem (i.e., the problem of deciding if a given definite Horn CNF represents a hypergraph Horn function) and the key realization problem (i.e.,  the problem of deciding if a given hypergraph is realized as the set of minimal keys of a hypergraph Horn function $h$) can be solved in polynomial time, and implicate sets can be generated with polynomial delay. We also discuss problems related to implicate-duality.  

\subsection{Core computation}

Let us start with the core computation, i.e., finding the largest implicate set contained in a given set. For a definite Horn CNF $\Psi$, we define 
\[\IS^1_\Psi(S) ~=~ \{v\in S\mid v \in \FC_\Psi(S-v)\}. 
\]

\begin{algorithm}[ht!]
\caption{\textsc{Core Computation}\label{alg:core}}
\begin{algorithmic}[1]
\Statex \textbf{Input}: A definite Horn CNF $\Psi$ and a set $S \subseteq V$.
\Statex \textbf{Output}: Core  $\IS_\Psi(S)$ of $S$ .
\State \textbf{Initialize}: $\IS^0_\Psi(S)\coloneqq S$
\For{$k=1,2, \dots , |S|$}
\State Compute $\IS^k_\Psi(S)=\IS^1_\Psi(\IS^{k-1}_\Psi(S))$ 
\EndFor
\State \textbf{return} $\IS^{|S|}_\Psi(S)$
\end{algorithmic}
\end{algorithm}

The following theorem shows that Algorithm {\sc Core Computation} computes the core in polynomial time. 

\begin{theorem}\label{l-IS-computation}
For a definite Horn CNF $\Psi$  and a subset $S\subseteq V$,  core $\IS_f(S)$ of $S$ can be obtained in $O(|S|^2{\cdot}\|\Psi\|)$ time. 
\end{theorem}
\begin{proof}
By definition, we have $\IS^1_\Psi(S)\subseteq S$, and note that equality holds if and only if $S$ is an implicate set. We also note that  $\IS^k_\Psi(S)=\IS^{k-1}_\Psi(S)$ implies $\IS^\ell_\Psi(S)=\IS^{k-1}_\Psi(S)$ for all $\ell \geq k-1$. Since $\emptyset$ is an implicate set for any Boolean function, Algorithm~\ref{alg:core} always outputs an implicate set of $\Psi$ contained in $S$. Moreover, for any implicate set  $I\subseteq S$  of $\Psi$, all variables $v\in I$ satisfy  $v \in \FC_f(S-v)$, since $I-v\subseteq S-v$. Thus, the set generated by the algorithm contains all implicate sets $I$ contained in $S$. This implies that the algorithm computes the core of $S$ correctly. Since for a subset $X\subseteq V$ and variable $v\in X$ we can test the condition $v \in \FC_\Psi(X-v)$ in $O(\|\Psi\|)$ time, the theorem follows. 
\end{proof}

\subsection{Hypergraph Horn recognition problem}
\label{sec-recog}

We now move to the recognition problem of hypergraph Horn functions. The following algorithm uses Algorithm~\ref{alg:core} as a subroutine polynomially many times.

\begin{algorithm}[ht!]
\caption{\textsc{Hypergraph Horn Recognition}\label{alg:Horn-detection}}
\begin{algorithmic}[1]
\Statex \textbf{Input}: A definite Horn CNF $\Psi$.
\Statex \textbf{Output}: A hypergraph $\cH\subseteq 2^V$ such that $\Psi=\Phi_{\cH}$ if $\Psi$ represents a hypergraph Horn function, and  ``\textsc{No}" otherwise.
\State \textbf{Initialize}: $\cH\coloneqq \emptyset$
\For{\textbf{each} clause $A\to v$ of $\Psi$}
\While{$A\overset{\Phi_\cH}{\nrightarrow}v$}\label{line4}
\State $T\coloneqq  \FC_{\Phi_\cH}(A)$\label{line5}
\For{\textbf{each} $u \in V\setminus T$}
\If{$u\in  \IS_\Psi(T+u)$}\label{line6}
\State $\cH\coloneqq  \cH\cup \{ \IS_\Psi(T+u)\}$\label{line7}
\EndIf
\EndFor
\If{$T=\FC_{\Phi_\cH}(A)$}\label{algo-line-1a}
\State \textbf{return} ``No"\label{line10}
\EndIf
\EndWhile
\EndFor
\State \textbf{return} $\cH$\label{line14}
\end{algorithmic}
\end{algorithm}

\begin{theorem}\label{t-Q1-solution}
For a definite Horn CNF $\Psi$, we can decide in $O(|V|^4{\cdot}|\Psi|{\cdot}\|\Psi\|)$ time if $\Psi$ represents a hypergraph Horn function, and if yes, construct a hypergraph $\cH\subseteq 2^V$ such that $h=\Phi_\cH$ and $|\cH|\leq |V|{\cdot}|\Psi|$. 
\end{theorem}
\begin{proof}
We show that Algorithm~\ref{alg:Horn-detection} satisfies the statement of the theorem. The algorithm makes use of  Characterization \ref{it:d} of Theorem~\ref{t-MAIN-hH} and the efficient computation of the cores by Algorithm~\ref{alg:core}. 

Let us observe first that the algorithm  collects in $\cH$ implicate sets of $\Psi$ (see lines~\ref{line6}-\ref{line7}), and thus  $\Psi\leq \Phi$ always holds during the iterations of the algorithm. We also note that $T$ computed in line~\ref{line5} is a false set of $\Psi$, since $A\overset{\Psi}{\to} v$, $A \subseteq T$, and $v \not\in T$. Thus, if the condition in line~\ref{line6} holds for none of the elements $u \in V \setminus T$, then by Theorem~\ref{t-MAIN-hH} \ref{it:d}, $\Psi$ does not represent a hypergraph Horn function, and the algorithm halts correctly in line~\ref{line10}. 

On the other hand, if the condition in line~\ref{line6} holds, then we get an implicate set $I$ of $\Psi$, and by adding this implicate set to $\cH$ changes $\Phi_\cH$ such that the set $T$ recomputed on line~\ref{line5} strictly increases in size. Thus the while-loop will be executed for every clause of $\Psi$ at most $|V|$ times, and if the while-loop for clause $A\to v$ does not stop in line~\ref{line10}, then the clause becomes an implicate of $\Phi_\cH$ at the end. Consequently,  we have $\Psi=\Phi_\cH$, if the algorithm halts in line~\ref{line14}. Moreover, we have $|\cH|\leq |V|{\cdot}|\Psi|$ upon termination, since $\cH$ is updated at most $|V|{\cdot}|\Psi|$ times.  
	
As for the time complexity, we can see that the core $\IS_\Psi$ in line~\ref{line6} and the closure $\FC_{\Phi_\cH}$ in lines~\ref{line5} and~\ref{algo-line-1a} can be computed at  most $|V|^2{\cdot}|\Psi|$ and $2{\cdot}|V|{\cdot}|\Psi|$ times, respectively. Therefore, the algorithm requires $O(|V|^4{\cdot}|\Psi|{\cdot}\|\Psi\|)$ time by Theorem~\ref{l-IS-computation} and $\|\Phi\|\leq |V|^2{\cdot}|\Psi|$. 
\end{proof}

\subsection{Key realization}

In the key realization problem we are given a hypergraph $\cK \subseteq 2^V$, and the goal is to decide if $\cK$ can be realized as the set of minimal keys of a hypergraph Horn function $h$, i.e., $\cK=\cK(h)$. Note that the \emph{key set} (i.e., family of all minimal keys) of  a definite Horn function is Sperner and any Sperner hypergraph is realized as the key set of a definite Horn function. Thus we may assume that $\cK$ is Sperner. 

Let us first review basic properties of minimal keys and maximal nontrivial true sets of definite Horn functions. For a Sperner hypergraph $\cK\subseteq 2^V$, let  $\cK^+$ denote the family of supersets of the hyperedges of $\cK$, i.e., 
\[\cK^+=\{ W \subseteq V \mid W \supseteq K \mbox{ for some } K \in \cK\}, 
\]
and we denote by $\cM=\cK^{dc}$ the associated family of maximal independent sets of $\cK$. Recall that if $\cK=\cK(h)$ for a definite Horn function $h$, then we have $\cM=\cM(h)$ by Lemma~\ref{l-E-M(h)-K(h)}.

\begin{lemma}
\label{lemma-xop1}
A Sperner hypergraph $\cK$ is realized as a key set of a definite Horn function $h$ if and only if 
\begin{enumerate}[label={\rm (\roman*)}]\itemsep0em
\item \label{xop1:i} $\cM \subseteq \cT(h)$, and 
\item \label{xop1:ii} $\cK^+\setminus \{V\}  \subseteq \cF(h)$. 
\end{enumerate}
\end{lemma}
\begin{proof}
The statement follows from the definitions of $\cK(h)$ and $\cM(h)$ and their relation $\cM(h) = \cK(h)^{dc}$ by Lemma~\ref{l-E-M(h)-K(h)}.
\end{proof}

We restate this lemma for the key realization problem in more constructive terms. We call a subset $I\subseteq V$ a \emph{potential implicate set} for a hypergraph $\cK$ if $\cM\subseteq \cT(\Phi_{\{ I\}})$. By Lemma~\ref{lemma-xop1}\ref{xop1:i},  any implicate set of a definite Horn function $h$ with $\cK(h)=\cK$ must be a potential implicate set, since $\cM=\cK(h)^{dc}=\cM(h)\subseteq \cT(h)$. Let $\cP(\cK) \subseteq 2^V$ denote the \emph{family of all potential implicate sets} for a Sperner hypergraph $\cK$. As $\cT(\Phi_I)\cup\cT(\Phi_J)\subseteq \cT(\Phi_{I\cup J})$ holds for any pair $I,J\subseteq V$ of subsets, $\cP(\cK)$ is closed under taking union.

\begin{lemma}
\label{lemma-xop2}
A Sperner hypergraph $\cK$ is realized as a key set of a hypergraph Horn function $h$ if and only if 
\begin{enumerate}[label={\rm (\roman*)}]\itemsep0em
\item $\cI(h) \subseteq \cP(\cK)$, and 
\item for any $F \in \cK^+\setminus \{V\}$ there exists $I\in \cI(h)$ with $|I \setminus F|=1$. 
\end{enumerate}
\end{lemma}
\begin{proof}
The statement follows from ~\eqref{e-true-set-of-h} and Lemma~\ref{lemma-xop1}.
\end{proof}

Clearly, having more implicate sets makes it easier to satisfy Lemma~\ref{lemma-xop2} (ii), which implies the following lemma. 

\begin{lemma}\label{l-key-Horn-iff}
A Sperner hypergraph $\cK$ is realized as the  key set of a definite Horn function  if and only if the hypergraph Horn function $h=\Phi_{\cP(\cK)}$ realizes $\cK$ as its key set. Furthermore, $\Phi_{\cP(\cK)}$ is the unique maximal circular minorant of the hypergraph Horn functions that have $\cK$ as their key sets, i.e., any hypergraph Horn function $f$ such that $\cK=\cK(f)$ satisfies $f \geq \Phi_{\cP(\cK)}$. 
\end{lemma}
\begin{proof}
Let $h$ be a hypergraph Horn function represented by $\Phi_{\cP(\cK)}$. Then $h$ satisfies \ref{xop1:i} of Lemma~\ref{lemma-xop1} by the definition of potential implicate sets.  Furthermore, we have  $h \leq f$ for any hypergraph Horn function $f$ satisfying Lemma~\ref{lemma-xop1}\ref{xop1:i}. This together with Lemma~\ref{lemma-xop1}\ref{xop1:ii} completes the proof. 
\end{proof}

In order to compute a hypergraph Horn function that realizes $\cK$ as its key set, let us first consider potential implicate sets. The following characterization shows that it can be checked in polynomial time whether a given a subset is in $\cP(\cK)$ or not. 
	
\begin{lemma}\label{l-pimp-iff}
A subset $I \subseteq V$ is a potential implicate set for a hypergraph $\cK$ if and only if 
\begin{equation}\label{e-pimp}
\text{for every}\ K\in\cK\ \text{and}\ u\in K\cap I\ \text{there exists}\ K'\in\cK \text{ such that } K'\subseteq (K\cup I) - u.
\end{equation}
\end{lemma}
\begin{proof}
We first assume that $I$ is not  a potential implicate set for $\cK$, i.e., $\Phi_{\{ I\}}(S)=0$ for some $S \in \cM$. Then there exists a variable $u$ such that $I \setminus S=\{u\}$. By the definition of $\cM=\cK^{dc}$, there exists a set $K \in \cK$ such that $K - u \subseteq S$ and $u \in K$. For these $I$, $u$, and $K$, we have $(K \cup I)-u \subseteq S$, which implies that no $K' \in \cK$ satisfies the inclusion in  (\ref{e-pimp}) since $S\in\cK^{dc}$.  

To show the other direction, let us assume that no set in $\cK$ is contained in $(K \cup I )-u$ for some $K \in \cK$ and $u \in K \cap I$. Then by $\cM=\cK^{dc}$, some $S \in \cM$ contains $ (K \cup I )-u$. This implies  that $\Phi_{\{ I\}}(S)=0$, since $I \setminus S=\{u\}$, and thus the claim follows. \end{proof}

By definition, condition \eqref{e-pimp} only deals with sets $K \in \cK$ with $K \cap I\not=\emptyset$. We note that condition \eqref{e-pimp} can be checked in $O(|I|+\|\cK\|^2)$ time for a given set $I\subseteq V$, where $\|\cK\|=\sum_{K\in \cK}|K|$.  

\begin{lemma}\label{unique-max-pis}
Let $\cK$ be a Sperner hypergraph. For any set $S \subseteq V$, there exists a unique maximal potential implicate set contained in $S$. 
\end{lemma}
\begin{proof}
Since $\emptyset$ is a potential implicate set for $\cK$, the lemma follows by $\cP(\cK)$ being closed under taking union.
\end{proof}

Let us now describe an algorithm for computing such  potential implicate sets. For a subset $S \subseteq V$, let $\PP_\cK(S)$ denote the maximal potential implicate set for $\cK$ included in $S$, and define
\[\PP^1_\cK(S) ~=~ \{u\in S\mid (S+K)-u \in \cK^+ \mbox{ for all } K \in \cK \mbox{ with } u\in K\}.  \]
By definition, a subset $I\subseteq V$ is a potential implicate set for $\cK$ if and only if $I=\PP_\cK(I)=\PP^1_\cK(I)$. 

\begin{algorithm}[ht!]
\caption{\textsc{Potential Implicate Set Computation}\label{alg:pis}}
\begin{algorithmic}[1]
\Statex \textbf{Input}: A Sperner hypergraph $\cK\subseteq 2^V$ and a set $S \subseteq V$.
\Statex \textbf{Output}: A unique maximal potential implicate set for $\cK$  contained in $S$.
\State \textbf{Initialize}: $\PP^0_\cK(S)\coloneqq S$
\For{$k=1,2, \dots , |S|$}
\State Compute $\PP^k_\cK(S)=\PP^1_\cK(\PP^{k-1}_\cK(S))$ 
\EndFor
\State \textbf{return} $\PP_\cK(S)=\PP^{|S|}_\cK(S)$
\end{algorithmic}
\end{algorithm}

\begin{lemma}\label{um-pis-computation}
For a Sperner hypergraph $\cK \subseteq 2^V$ and a set $S \subseteq V$, the unique maximal potential implicate set for $\cK$ contained in $S$ can be obtained in $O(|S|{\cdot}\|\cK\|^2)$ time. 
\end{lemma}
\begin{proof}
Similarly to the operator $\IS$, we have $\PP^1_\cK(S)\subseteq S$, where equality holds if and only if $S$ is a potential implicate set for $\cK$. Furthermore, $\PP^k_\cK(S)=\PP^{k-1}_\cK(S)$ implies that $\PP^\ell_\cK(S)=\PP^{k-1}_\cK(S)$ for all $\ell \geq k-1$. Since $\emptyset$ is a potential implicate set, it follows from Lemma~\ref{unique-max-pis} that Algorithm~\ref{alg:pis} outputs the unique maximal potential implicate set for $\cK$ contained in $S$.  

As for the time complexity, let $A=S\setminus (\bigcup_{K \in \cK}K)$ and $B=S-A$, and consider the slightly different version of Algorithm~\ref{alg:pis} which is obtained by first replacing $S$ by $B$ and then modifying the output by $\PP^{|B|}_\cK(B) \cup A$. Then it is not difficult to see that this modification still computes a desired potential implicate set and each for-loop can be executed in $O(\|\cK\|^2)$ time. Therefore, in total, the algorithm requires $O(|S|{\cdot}\|\cK\|^2)$ time, which completes the proof.  
\end{proof}

By Lemma~\ref{l-key-Horn-iff}, the key realization problem can be solved by generating all potential implicate sets. However, in general, $|\cP(\cK)|$ may be exponentially larger than $|\cK|$ and $|V|$. Algorithm~\ref{alg:key} has similar characteristics to the hypergraph Horn recognition algorithm, and constructs only polynomially many potential implicate sets for $\cK$ that satisfy Lemma~\ref{lemma-xop1}\ref{xop1:ii}. 

\begin{algorithm}[ht!]
\caption{\textsc{Key Realization}\label{alg:key}}
\begin{algorithmic}[1]
\Statex \textbf{Input}: A Sperner hypergraph $\cK \subseteq 2^V$.
\Statex \textbf{Output}: A hypergraph $\cH\subseteq 2^V$ with $\cK(\Phi_\cH)=\cK$ if $\cK$ is realizable as the key set of a hypergraph Horn function, and ``\textsc{No}" otherwise.
\State \textbf{Initialize}: $\cH\coloneqq \emptyset$
\For{\textbf{each}  $J$ in  $\cK$}
\While{$\FC_{\Phi_\cH}(J)\not=V$}\label{line4---a}
\State $T\coloneqq  \FC_{\Phi_\cH}(J)$
\For{\textbf{each} $v \in V \setminus T$}
\If{$v\in  \PP_\cK(T+v)$}\label{line6q}
\State $\cH\coloneqq  \cH\cup \{ \PP_\cK(T+v)\}$\label{line7q}
\EndIf
\EndFor
\If{$T=\FC_{\Phi_\cH}(J)$}
\State \textbf{return} ``\textsc{No}"\label{line10q}
\EndIf
\EndWhile
\EndFor
\State \textbf{return} $\cH$\label{line14q}
\end{algorithmic}
\end{algorithm}

\begin{theorem}\label{t-comp-keyrealization}
For a Sperner hypergraph $\cK \subseteq 2^V$, we can decide in $O(|V|^3{\cdot}|\cK|{\cdot}\|\cK\|^2)$ time if $\cK$ is realizable as the key set of a hypergraph Horn function, and if yes, construct a hypergraph $\cH\subseteq 2^V$ such that $\cK=\cK(\Phi_\cH)$ and $|\cH|\leq |V|{\cdot}|\cK|$. 
\end{theorem}
\begin{proof}
Algorithm~\ref{alg:key} satisfies the statement of the theorem; the proof is analogous to that of Theorem~\ref{t-Q1-solution}.  
\end{proof}

\subsection{Implicate set generation}

Finally, we consider the problem of generating implicate sets of a given definite Horn CNF. Note that  a definite Horn function might have exponentially many implicate sets in the size of its CNF representation (see e.g., Example~\ref{ex-is-beyond-union}).  It is customary to measure the complexity of such generation problems in both the size of the input and output, see e.g.,~\cite{LLK1980,Val1979}. A generation algorithm is called \emph{polynomial delay} if the time intervals between the start and the first output, between consecutive outputs, and between the last output and the halt are all bounded by a polynomial in the input size. 

In order to have such an algorithm for generating implicate sets, we apply the so-called flashlight search discussed e.g., in~\cite{BGKM2001}, where flashlight search explores the space of all possible solutions by checking if the current partial solution can be extended to a complete solution. More precisely, our algorithm solves the following problem as a subroutine. 
 
\decprob{Implicate Set Extension($\Psi, X,Y$)}{A definite Horn CNF $\Psi$ and disjoint subsets $X$ and $Y$ in $V$.}{Does $\Psi$ has an implicate set $I$ such that $X \subseteq I$ and $I \cap Y=\emptyset$?}  

\begin{lemma}\label{t-Q3-solution}
Let $X$ and $Y$ be disjoint subsets in $V$, and let $\Psi$ be a definite Horn CNF. Then $\Psi$ has an implicate set $I$ such that $X \subseteq Y$ and $I \cap Y=\emptyset$  if and only if $\IS_\Psi(V\setminus Y)\supseteq X$. Hence {\sc Implicate Set Extension}$(\Psi, X,Y)$ can be solved in $O(|V|^2{\cdot}\|\Psi\|)$ time. 
\end{lemma}
\begin{proof}
The claim follows by the fact that $\IS_h(V\setminus Y)$ is the unique maximal implicate set of $h$ that is disjoint from $Y$. The complexity  of Problem {\sc  Implicate Set Extension}$(\Psi, X,Y)$ follows by Theorem~\ref{l-IS-computation}.
\end{proof}

Therefore we get the following theorem. 

\begin{theorem}\label{t-Q2-solution}
For a given definite Horn CNF $\Psi$, we can generate implicate sets of $\Psi$ in polynomial delay. 
\end{theorem}
\begin{proof}
The claim follows from Lemma~\ref{t-Q3-solution} and the flashlight search, see e.g.,~\cite{BGKM2001}.
\end{proof}

Note that $\cI(\Psi)$ is closed under union and $\IS_h(V\setminus Y)$ is the unique maximal implicate set of $h$ that is disjoint from $Y$. Hence, by utilizing such properties, implicate sets can be generated faster than the standard flashlight search, although we do not discuss it in this paper. 

\subsection{Checking implicate-duality}\label{subsec: i-duality}

Deciding the computational complexity of checking implicate-duality remains an intriguing open question. In this problem, we are given two definite Horn CNFs $\Psi$ and $\Gamma$, and the goal is to decide if $\Psi=\Gamma^i$. The problem is open even when $\Psi=\Gamma$, that is, one would like to check if a definite Horn CNF $\Psi$ represents a self $i$-dual Horn function or not. We call a definite Horn function $h$ \textit{self $i$-dual} if $h^i=h$. We know by Lemma \ref{t-f^i-is-hH} that $h$ must be hypergraph Horn in this case.  

As initial steps toward understanding these problems, we show that $\Psi \geq \Gamma^i$ can be checked in polynomial time and that the problem of deciding $\Psi=\Psi^i$ belongs to co-NP.  

\begin{theorem}\label{t-Q4-solution}
Let $\Psi$ and $\Gamma$ be two definite Horn CNFs. Then $\Psi\geq \Gamma^i$ can be checked in $O(|V|^2{\cdot}|\Psi|{\cdot}\|\Gamma\|)$ time. 
\end{theorem}
\begin{proof}
Note that  $\Psi\not\geq \Gamma^i$ holds if and only if there exists a false set $F$ of $\Psi$ such that $V\setminus F$ is an implicate set of $\Gamma$. This is equivalent to the existence of a clause $A \to v$ in $\Psi$ such that the output of {\sc Implicate Set Extension}($\Gamma, \{v\},A$) is ``Yes''. Therefore, by Lemma~\ref{t-Q3-solution}, $\Psi\geq \Gamma^i$ can be checked in $O(|V|^2{\cdot}|\Psi|{\cdot}\|\Gamma\|)$ time. 
\end{proof}

We show next that self implicate duality of $h$ can be characterized as a property of its complete implicate set $\cI(h)$. This characterization then allows us to show that the problem of deciding $h=h^i$ for a function $h$ represented by definite Horn CNF belongs to co-NP.

\begin{theorem}
\label{t-self-i-dual}
A definite Horn function $h$ is self $i$-dual if the hypergraph  $\cH=\cI(h)$ satisfies the following properties.
\begin{enumerate}[label={\rm (\alph*)}]\itemsep0em
    \item \label{idual:a} $|H\cap H'|\neq 1$ for all $H,H'\in\cH$.
    \item \label{idual:b} $\cH$ is a maximal subfamily of $2^V$ with respect to property \ref{idual:a}.
\end{enumerate}
\end{theorem}
Let us remark that since $H=H'$ is possible in \ref{idual:a}, $|H|\geq 2$ follows for all $H\in\cI(h)$ for a self $i$-dual Horn function $h$.

Before proving this theorem, we need a technical lemma that maybe of interest on its own.

\begin{lemma}\label{l-self-i-dual-majorant}
Given a hypergraph $\cH\subseteq 2^V$, we have $\Phi_\cH \geq \Phi^i_\cH$ if and only if $\cH$ satisfies condition \ref{idual:a} of Theorem \ref{t-self-i-dual}.
\end{lemma}
\begin{proof}
Let us introduce $h=\Phi_\cH$. By definition $\cT(h^i)=\cI(h)^c$. Thus we have $\cT(h^i)\subseteq \cT(h)$ if and only if $\cI(h)^c\subseteq \cT(h)$. By Lemma \ref{l-E-true-set-implicate-set-basics} we have that a subset $S\subseteq V$ is a true set of $h$ if and only if there exists no implicate set $I\in\cI(h)$ such that $|I\setminus S|=1$. Thus, we have $\cI(h)^c\subseteq \cT(h)$ if and only if for all $S\in\cI(h)^c$ and $I\in\cI(h)$ we have $|I\setminus S|\neq 1$, or equivalently, for all $I,I'\in\cI(h)$ we have $|I\cap I'|\neq 1$. 
	
We claim that this latter condition holds if and only if $\cH$ satisfies property \ref{idual:a} of Theorem \ref{t-self-i-dual}. Since $\cH\subseteq \cI(h)$, the necessity of condition \ref{idual:a} is implied. For the sufficiency, assume for a contradiction that there exist implicate sets $I,I'\in\cI(h)$ such that $|I\cap I'|=1$, i.e., $|I\setminus (V\setminus I')|=1$. This relation then implies that $\FC_h(V\setminus I')\neq (V\setminus I')$. Since $\FC_h=\FC_{\Phi_\cH}$, there exists a hyperedge $H\in \cH$ such that $|H\setminus (V\setminus I')|=1$, or equivalently, $|H\cap I'|=1$. We can write this equivalently, as $|I'\setminus (V\setminus H)|=1$ from which $\FC_h(V\setminus H)\neq (V\setminus H)$ follows, implying the existence of an $H'\in\cH$ such that $|H'\setminus (V\setminus H)|=1$. 
Thus these $H$ and $H'$ would violate condition \ref{idual:a},  completing the sufficiency.  
\end{proof}

\begin{proof}[Proof of Theorem \ref{t-self-i-dual}]
Since $\Phi_\cH\geq \Phi^i_\cH$ must hold for self-$i$-dual functions, property \ref{idual:a} of  Theorem \ref{t-self-i-dual} is necessary by Lemma \ref{l-self-i-dual-majorant}. Assume first that $\cH$ satisfies that $\Phi_\cH > \Phi_\cH^i$. Then  we have a subset $S\subseteq V$ such that $\Phi_\cH^i(S)=0$ and $\Phi_\cH(S)=1$. The equality $\Phi_\cH^i(S)=0$ implies $(V\setminus S)\notin \cH$ by the definition of $i$-duality. The equality $\Phi_\cH(S)=1$ implies that for all $H\in \cH$ we have $|H\setminus S|\neq 1$, or equivalently that $|H\cap (V\setminus S)|\neq 1$. 
This implies that $\cH \cup \{V \setminus S\}$ satisfies \ref{idual:b}, that is, $\cH$ is not  maximal with respect to \ref{idual:a}.

For the converse direction, assume that $\cH$ satisfies \ref{idual:a}, but not maximal with respect to \ref{idual:a}, i.e., there exists a subset $S\subseteq V$ such that $S\notin \cH$ and $|S\cap H|\neq 1$ for all $H\in \cH$. 
The latter condition implies $\Phi_\cH(V\setminus S)=1$ by Lemma \ref{l-E-true-set-implicate-set-basics}.
On the other hand, since $\cH=\cI(h)$, 
we have $\cT(\Phi_\cH^i)=\cH^c$. 
Thus $S\notin \cH$ implies that $\Phi_\cH^i(V\setminus S)=0$, 
which concludes that 
$\Phi_\cH > \Phi^i_\cH$.  
\end{proof}

The characterization given above is analogous to the characterization of monotone duality, see \cite{bioch1995complexity} for further details.

\begin{corollary}\label{c-self-i-dual-co-NP}
The decision problem $h=h^i$ for functions $h$ represented by definite Horn CNFs belongs to co-NP.
\end{corollary}
\begin{proof}
To prove $h\neq h^i$ it is enough by Theorem \ref{t-self-i-dual} to present either two implicate sets $I,I'\in\cI(h)$ such that $|I\cap I'|=1$ or a subset $S\subseteq V$ such that $S\notin\cI(h)$ and $|S\cap I|\neq 1$ for all $I\in\cI(h)$.

In the first case we can test memberships $I,I'\in \cI(h)$ in polynomial time using any Horn CNF representation of $h$. 

In the second case we can also test $S\notin\cI(h)$ in polynomial time using the Horn CNF representing $h$. Furthermore, we can test for all $v\in S$ if $v\in I=\IS_h((V\setminus S)+v)$ in polynomial time by Theorem \ref{l-IS-computation}. If the answer is ``No'' for all $v\in S$, then we must have $|S\cap I|\neq 1$ for all $I\in\cI(h)$. 
\end{proof}

\section{Conclusions and open problems}

In the present paper, we introduced the notion of hypergraph Horn functions, a subclass of Boolean functions with rich structural properties that generalizes matroids and equivalence relations. We gave several equivalent characterizations of hypergraph Horn functions, introduced a new duality that generalizes matroid duality, and studied algorithmic problems including the recognition problem, key realization, and implicate set generation. 

We conclude our paper with listing a few open problems. Assume that $\Psi$ and $\Gamma$ are definite Horn CNFs, and all questions are about polynomial time computability.

\begin{qu}
Does $\Psi \leq \Gamma^i$ hold?
\end{qu}

\begin{qu}
Does $\Psi=\Psi^i$ hold?
\end{qu}

\begin{qu}
Is there a Sperner hypergraph $\cH\subseteq 2^V$ such that $\Psi=\Phi_\cH$?
\end{qu}

\medskip
\paragraph{Acknowledgements.} This  work  was  supported  by  the  Research  Institute  for  Mathematical  Sciences,  an  International Joint Usage/Research Center located in Kyoto University. The first author was supported by the Lend\"ulet Program of the Hungarian Academy of Sciences -- grant number LP2021-1/2021 and by the Hungarian National Research, Development and Innovation Office -- NKFIH, grant number FK128673. The third author was supported by JSPS KAKENHI Grant Numbers JP20H05967, JP19K22841, and JP20H00609.

\bibliographystyle{abbrv}
\bibliography{hypergraph_horn} 

\begin{thebibliography}{10}

\bibitem{adaricheva_et_al:DR:2014:4619}
K.~V. Adaricheva, G.~F. Italiano, H.~K. B{\"u}ning, and G.~Tur{\'a}n.
\newblock {Horn formulas, directed hypergraphs, lattices and closure systems:
  related formalisms and applications (Dagstuhl Seminar 14201)}.
\newblock {\em Dagstuhl Reports}, 4(5):1--26, 2014.

\bibitem{matroid_horn}
K.~B{\'e}rczi, E.~Boros, and K.~Makino.
\newblock Matroid {Horn} functions.
\newblock {\em arXiv preprint arXiv:2301.06642}, 2023.

\bibitem{bioch1995complexity}
J.~C. Bioch and T.~Ibaraki.
\newblock Complexity of identification and dualization of positive boolean
  functions.
\newblock {\em Information and Computation}, 123(1):50--63, 1995.

\bibitem{BGKM2001}
E.~Boros, V.~Gurvich, L.~Khachiyan, and K.~Makino.
\newblock Dual-bounded generating problems: Partial and multiple transversals
  of a hypergraph.
\newblock {\em SIAM Journal on Computing}, 30(6):2036--2050, 2001.

\bibitem{CLM81}
A.~Chandra, H.~Lewis, and J.~Makowsky.
\newblock Embedded implicational dependencies and their inference problem.
\newblock In {\em ACM STOCS}, pages 342--354, 1981.

\bibitem{CL73}
C.~Chang and R.~Lee.
\newblock {\em Symbolic Logic and Mechanical Theorem Proving}.
\newblock Academic Press, New York, 1973.

\bibitem{CH11}
Y.~Crama and P.~L. Hammer.
\newblock {\em Boolean functions: Theory, algorithms, and applications}.
\newblock Cambridge University Press, 2011.

\bibitem{Fag82}
R.~Fagin.
\newblock {H}orn clauses and database dependencies.
\newblock {\em J. Assoc. Comput. Machinery}, 29:952--985, 1982.

\bibitem{Hod93}
W.~Hodges.
\newblock {\em Logical features of {H}orn clauses}, volume~1 of {\em Handbook
  of Logic in Artificial Intelligence and Logic Programming}, pages 449--503.
\newblock Oxford University Press, 1993.

\bibitem{Horn1951}
A.~Horn.
\newblock On sentences which are true of direct unions of algebras.
\newblock {\em Journal of Symbolic Logic}, 16:14–21, 1951.

\bibitem{IM87}
A.~Itai and J.~Makowsky.
\newblock Unification as a complexity measure for logic programming.
\newblock {\em Journal of Logic Programming}, pages 105--117, 1987.

\bibitem{KhardonRoth1996}
R.~Khardon and D.~Roth.
\newblock Reasoning with models.
\newblock {\em Artificial Intelligence}, 87:187–213, 1996.

\bibitem{LLK1980}
E.~L. Lawler, J.~K. Lenstra, and A.~H. G.~R. Kan.
\newblock Generating all maximal independent sets: Np-hardness and
  polynomial-time algorithms.
\newblock {\em SIAM Journal on Computing}, 9:558–565, 1980.

\bibitem{Mai83}
D.~Maier.
\newblock {\em The theory of relational databases}.
\newblock Computer Science Press, Rockville, MD, 1983.

\bibitem{MW88}
D.~Maier and D.~Warren.
\newblock {\em Computing with Logic: Logic Programming with {PROLOG}}.
\newblock Benjamin/Cummings Publ. Co.,, Menlo Park, CA, 1988.

\bibitem{Mak87}
J.~Makowsky.
\newblock Why {H}orn formulas matter in computer science: initial structures
  and generic examples.
\newblock {\em J. Comput. System Sci.}, 34:266--292, 1987.

\bibitem{McKinsey1943}
J.~McKinsey.
\newblock The decision problem for some classes of sentences without
  classifiers.
\newblock {\em Journal of Symbolic Logic}, 8:61–76, 1943.

\bibitem{nishimura2009lost}
H.~Nishimura and S.~Kuroda.
\newblock {\em A Lost Mathematician, {T}akeo {N}akasawa: The Forgotten Father
  of Matroid Theory}.
\newblock Springer Science \& Business Media, 2009.

\bibitem{SDPF81}
Y.~Sagiv, C.~Delobel, D.~Parker, and R.~Fagin.
\newblock An equivalence between relational database dependencies and a
  fragment of propositional logic.
\newblock {\em Journal of the ACM}, 28(3):435--453, 1981.

\bibitem{Tutte1971}
W.~Tutte.
\newblock {\em Introduction to the theory of matroids}, volume Modern Analytic
  and Computational Methods in Science and Mathematics, vol. 37.
\newblock American Elsevier Publishing Company, New York, 1971.

\bibitem{Val1979}
L.~Valiant.
\newblock The complexity of enumeration and reliability problems.
\newblock {\em SIAM Journal on Computing}, 8:410–421, 1979.

\bibitem{Welsh1976}
D.~J.~A. Welsh.
\newblock {\em Matroid Theory}, volume L.M.S. Monographs, vol. 8.
\newblock Academic Press, 1976.

\bibitem{whitney1992abstract}
H.~Whitney.
\newblock On the abstract properties of linear dependence.
\newblock In {\em Hassler Whitney Collected Papers}, pages 147--171. Springer,
  1992.

\end{thebibliography}

\end{document}